\documentclass[11pt]{article}

\usepackage{amsmath}
\usepackage{bm}
\usepackage{amssymb}
\usepackage{amstext}
\usepackage{latexsym}
\usepackage{color}
\usepackage{graphicx}
\usepackage{epsfig}
\usepackage{subfigure}
\usepackage{rotating}
\usepackage{curves}
\usepackage{epic}
\usepackage{amsthm}

\usepackage{amsfonts}
\usepackage[T1]{fontenc}
\usepackage[latin1]{inputenc}
\usepackage{authblk}
\usepackage{tikz}
\usetikzlibrary{shapes,arrows}
\usepackage{graphics}
\usepackage{verbatim}
\usepackage{color}
\usepackage{hyperref}

\newcommand{\be}{\begin{equation}}
\newcommand{\en}{\end{equation}}

\newtheorem{thm}{Theorem}
\newtheorem{cor}[thm]{Corollary}

\newtheorem{defi}{Definition}[section]
\newtheorem{lem}[defi]{Lemma}

\newtheorem{Theo}{Theorem}[section]

\newtheorem{Prop}[Theo]{Proposition}
\newtheorem{Cor}[Theo]{Corollary}
\newtheorem{remark}[Theo]{Remark}

\newcommand{\bedefin}{\begin{defi}}
\newcommand{\findefi}{\end{defi} \medskip}

\newcommand{\betheo}{\begin{theorem}$\!\!${\bf \,\,\,}}
\newcommand{\entheo}{\end{theorem}}
\newcommand{\enth}{\end{theorem}}

\newcommand{\becor}{\begin{cor}$\!\!${\bf .}}
\newcommand{\encor}{\end{cor}}

\newcommand{\belem}{\begin{lem}$\!\!${\bf }}
\newcommand{\enlem}{\end{lem}}

\newcommand{\bea}{\begin{eqnarray}}
\newcommand{\ena}{\end{eqnarray}}

\newcommand{\beano}{\begin{eqnarray*}}
\newcommand{\enano}{\end{eqnarray*}}

\newcommand{\bee}{\begin{enumerate}}
\newcommand{\ene}{\end{enumerate}}

\newcommand{\bei}{\begin{itemize}}
\newcommand{\eni}{\end{itemize}}

\newcommand{\betab}{\begin{tabular}}
\newcommand{\entab}{\end{tabular}}

\newcommand{\bd}{\begin{displaymath}}

\newcommand{\g}{G_{\hbox{\tiny{NC}}}}

\newcommand{\G}{\mathfrak{g}_{\hbox{\tiny{NC}}}}
\newcommand{\gh}{G_{\hbox{\tiny{H}}}}
\newcommand{\Gh}{\mathfrak{g}_{\hbox{\tiny{H}}}}

\catcode `\@=11 \@addtoreset{equation}{section}

\catcode `\@=12

\begin{document}

\title{Deformation of Noncommutative Quantum Mechanics}
\author[1]{Jian-Jian Jiang\thanks{j.j.jiang@foxmail.com}}
\author[2]{S. Hasibul Hassan Chowdhury\thanks{shhchowdhury@gmail.com}}
\affil[1,2]{Chern Institute of Mathematics, Nankai University, Tianjin 300071, P. R. China}
\date{}

\maketitle

\begin{abstract}
In this paper, the Lie group $\g^{\alpha,\beta,\gamma}$, of which the kinematical symmetry group $\g$ of noncommutative quantum mechanics (NCQM) is a special case due to fixed nonzero $\alpha$, $\beta$ and $\gamma$, is three-parameter deformation quantized using the method suggested by Ballesteros and Musso in \cite{BaMu}. A certain family of QUE algebras, corresponding to $\g^{\alpha,\beta,\gamma}$ with two of the deformation parameters approaching zero, is found to be in agreement with the existing results of the literature on quantum Heisenberg group. Finally, we dualize the underlying QUE algebra to obtain an expression for the underlying $*$-product between smooth functions on $\g^{\alpha,\beta,\gamma}$.
\end{abstract}

\section{Introduction}\label{sec:intro}
Noncommutative quantum mechanics (NCQM) is a vibrant field of research these days. In addition to the canonical position-momentum noncommutativity, it also demands for the noncommutativity between the two position coordinates and the two momenta coordinates respectively for a system of two degrees of freedom. The noncommutativity of the position coordinates was first proposed by H. S. Snyder \cite{Sny} in his quest for the quantized nature of space-time. Such a model of space-time in scales as small as Planck length is also proposed, among others by Doplicher et al. \cite{Doplicheretal} in order to avoid creation of microscopic black holes to the effect of losing the operational meaning of localization in space-time. The noncommutativity of the momenta coordinates, on the other hand, emerges if one introduces a constant background magnetic field to the underlying system of two degrees of freedom.

The defining group of a two-dimensional quantum mechanical system is the well-known Heisenberg group (denoted by $\gh$ in the sequel). What runs parallel to $\gh$ in two dimensions, is the triply extended group of translations of $\mathbb{R}^{4}$ (denoted as $\g$ in the sequel) in the context of NCQM. This nilpotent Lie group was first introduced in \cite{ncqmjmp} and was later studied extensively in terms of its unitary dual in \cite{ncqmjpa}. It was also shown in \cite{ncqmjpa} that the unitary dual of $\gh$ sits inside that of $\g$ pointing up the universal nature of $\g$ as the underlying group of quantum mechanics. This Lie group was later identified as the kinematical symmetry group for this model of NCQM in \cite{wigjmp} by computing its associated Wigner functions which under appropriate limit agrees with the quantum mechanical Wigner function originally computed by Wigner in his seminal paper \cite{wig}.

The well-established theory of deformation quantization insinuates the fact that the structure of a Lie group is not  quite a rigid one. Here one equips the enveloping algebra of the underlying Lie algebra with some additional structures in a compatible way to make it a Hopf algebra which is noncommutative (inherited from the Lie algebra structure) but co-commutative. For details on the Hopf algebraic terminologies, we refer the readers to the classic reading \cite{ChPr}. The goal of deformation quantization is to deform the coalgebra structure of the universal enveloping algebra considered as a Hopf algebra. The resulting deformed object known as quantized universal enveloping algebra (abbreviated as QUE algebra) still lies in the category of Hopf algebra. ``Quantum Group'' is a widely used misnomer in modern mathematical literature for QUE algebra resulting from the deformation of the universal enveloping algebra of some Lie algebra.

Quantum Heisenberg group in one and higher dimensions have been studied thoroughly in the past (see, for example, \cite{CeGiSoTa, Bonechietal, Hussinetal}). Deformation quantization of NCQM, on the other hand, had been studied recently by Bastos et al. in a series of papers by focussing on the computation of various $*$-products between functions lying in $\mathcal{S}(\mathbb{R}^{2n})$ (see, for example, \cite{bastosjmp, bastoscmp, Diasetal}). Here in this paper, following a method prescribed by Ballesteros et al. (see \cite{BaMu}), we take a different route to find a three-parameter deformation of the universal enveloping algebra associated with a family of Lie algebras denoted by $\G^{\alpha,\beta,\gamma}$ of which $\G$ (Lie algebra of $\g$) is a special case (see below with fixed nonzero $\alpha$, $\beta$ and $\gamma$). We verify that under classical limits of two of the deformation parameters, one of the families $\G^{\alpha,\beta,\gamma}$ due to $\beta=\gamma=0$ essentially reproduces the Hopf algebra structure associated with the Heisenberg group found by Celeghini et al. in \cite{CeGiSoTa}. We also find the $*$-products between smooth functions on $\g^{\alpha,\beta,\gamma}$ (Lie group corresponding to the Lie algebra $\G^{\alpha,\beta,\gamma}$) by dualising the underlying QUE algebra  and discuss the linear Poisson structures arising in the classical limit, i.e. when the deformation parameters approach zero. We verify that such linear Poisson structure indeed agrees with the underlying Lie bialgebra structure. 

The organisation of the paper is as follows. In section \ref{sec:deform}, following a brief description of the Lie group $\g^{\alpha,\beta,\gamma}$ to be studied, we find the QUE algebra associated with the underlying enveloping algebra and study various limiting cases to compare our results against the existing ones for the Heisenberg case. We also, compute the cocommutators associated with the generators from the explicit expressions of the coproducts of the underlying QUE algebra to find the corresponding bialgebra structure. In section \ref{sec:star-prod}, we compute the $*$-products between smooth functions on $\g$ by considering the dual QUE algebra and see how they give rise to linear Poisson structure that is in complete agreement with the bialgebra structure computed in section \ref{sec:deform}. Finally, in section \ref{sec:conclusions}, we give our concluding remarks and point to some possible future work.

\section{Deformation of $U(\G^{\alpha,\beta,\gamma})$}\label{sec:deform}
Let us first have a cursory look at the algebraic structure associated with the group $\g^{\alpha,\beta,\gamma}$. For any real numbers $\alpha,\beta,\gamma\in\mathbb{R}$, let $\g^{\alpha,\beta,\gamma}$ be a Lie group whose generic element is denoted by $(\theta,\phi,\psi,\mathbf{q},\mathbf{p})$, where $\mathbf{q},\mathbf{p}\in\mathbb{R}^2$. Then the group composition law for $\g^{\alpha,\beta,\gamma}$ reads (cf. \cite{ncqmjpa})
\begin{eqnarray*}
&(\theta,\phi,\psi,\mathbf{q},\mathbf{p})(\theta',\phi',\psi',\mathbf{q}',\mathbf{p}')= \qquad\qquad \\
&(\theta+\theta'+\frac{\alpha}{2}[\langle\mathbf{q},\mathbf{p}'\rangle-\langle\mathbf{p},\mathbf{q}'\rangle],\phi+\phi'+\frac{\beta}{2}[\mathbf{p}\wedge\mathbf{p}'],\psi+\psi'+\frac{\gamma}{2}[\mathbf{q}\wedge\mathbf{q}'],\mathbf{q}+\mathbf{q}',\mathbf{p}+\mathbf{p}').
\end{eqnarray*}
If $\mathbf{q}=(q_1,q_2)$, $\mathbf{p}=(p_1,p_2)$, then $\langle,\rangle$ and $\wedge$ are defined as $\langle\mathbf{q},\mathbf{p}\rangle=q_1 p_1+q_2 p_2$ and $\mathbf{q}\wedge\mathbf{p}=q_1 p_2-q_2 p_1$ respectively.

Note that if we denote the dimension of the position coordinate by $[q]$ and that of the momentum coordinate by $[p]$, then we immediately see that in order to have $\theta$, $\phi$ and $\psi$ to be all dimensionless, we must have 
$$[\alpha]=[\frac{1}{pq}], \qquad [\beta]=[\frac{1}{p^2}], \qquad [\gamma]=[\frac{1}{q^2}].$$
Let us denote the Lie algebra of $\g^{\alpha,\beta,\gamma}$ by $\G^{\alpha,\beta,\gamma}$. Since we are only intersted in the $\alpha\ne 0$ case, we may assume $\alpha\ne 0$ throughout this paper. Under this assumption, we can choose an appropriate basis of $\G^{\alpha,\beta,\gamma}$, namely, $\{\Theta, \Phi, \Psi, Q_1, Q_2, P_1, P_2\}$, such that the commutation relations between them have the following dimension-consistent form:
\begin{eqnarray*}
& [Q_i,P_j]=\frac{\delta_{ij}}{\alpha}\Theta, \quad [Q_1,Q_2]=\frac{\beta}{\alpha^2}\Phi, \quad [P_1,P_2]=\frac{\gamma}{\alpha^2}\Psi, \\
& [\Theta, \Phi]=[\Theta, \Psi]=[\Theta, Q_i]=[\Theta, P_i]=[\Phi, \Psi]=0, \\
& [\Phi, Q_i]=[\Phi, P_i]=[\Psi, Q_i]=[\Psi, P_i]=0, \quad (i,j=1,2),
\end{eqnarray*}
where $\delta_{ii}=1$ and $\delta_{ij}=0$ for $i\ne j$ is the Kronecker symbol.

Using the method suggested by A. Ballesteros and F. Musso \cite{BaMu}, we obtain a three-parameter deformation of the universal enveloping algebra $U(\G^{\alpha,\beta,\gamma})$, as stated in the next proposition. Note that our deformation is in the sense of ``formal deformation of Hopf algebras'', as is commonly adopted in the literature of ``Quantum Groups''. We refer the reader to the monograph \cite{ChPr} for a thorough introduction.

Throughout the paper we denote by $\mathbb{C}$ the set of complex numbers, and denote by $\mathbb{C}[[\hbar_1,\hbar_2,\hbar_3]]$ the commutative ring of formal power series in $\hbar_1$, $\hbar_2$ and $\hbar_3$ with complex coefficients. Then any associative algebra $U$ over $\mathbb C[[\hbar_1,\hbar_2,\hbar_3]]$ can be equipped with the $(\hbar_1,\hbar_2,\hbar_3)$-adic topology (see, e.g. \cite{ChPr}). Henceforth we may consider $U(\G^{\alpha,\beta,\gamma})$ as a Hopf algebra over $\mathbb{C}$.

\begin{Prop} Let $U_{\hbar_1,\hbar_2,\hbar_3}(\G^{\alpha,\beta,\gamma})$ be the algebra over $\mathbb C[[\hbar_1,\hbar_2,\hbar_3]]$ topologically generated by elements $\Theta,\Phi,\Psi,Q_1,Q_2,P_1,P_2$ and with the following defining relations:
\begin{eqnarray*}
&\left(\rho=\hbar_1 \Theta+\hbar_2 \Phi+\hbar_3 \Psi, \quad \lambda=\frac{\sinh{2\rho}}{2\rho}=\frac{e^{2\rho}-e^{-2\rho}}{4\rho}=1+\frac{2}{3}\rho^2+\cdots\right) \\
& [Q_i,P_j]=\delta_{ij}\frac{\lambda}{\alpha}\Theta, \qquad [Q_1,Q_2]=\frac{\beta\lambda}{\alpha^2}\Phi, \qquad [P_1,P_2]=\frac{\gamma\lambda}{\alpha^2}\Psi, \\
& [\Theta, \Phi]=[\Theta, \Psi]=[\Theta, Q_i]=[\Theta, P_i]=[\Phi, \Psi]=0, \\
& [\Phi, Q_i]=[\Phi, P_i]=[\Psi, Q_i]=[\Psi, P_i]=0, \qquad (i,j=1,2).
\end{eqnarray*}
Then $U_{\hbar_1,\hbar_2,\hbar_3}(\G^{\alpha,\beta,\gamma})$ is a topological Hopf algebra over $\mathbb C[[\hbar_1,\hbar_2,\hbar_3]]$ with coproduct defined by
\begin{eqnarray*}
&\Delta(\rho)=\rho\otimes 1+1\otimes\rho, \qquad \Delta(e^\rho)=e^\rho\otimes e^\rho, \\
&\Delta(\Theta)=\frac{\lambda\Theta\otimes e^{2\rho}+e^{-2\rho}\otimes\lambda\Theta}{\Delta(\lambda)}, \quad \Delta(\Phi)=\frac{\lambda\Phi\otimes e^{2\rho}+e^{-2\rho}\otimes\lambda\Phi}{\Delta(\lambda)}, \quad \Delta(\Psi)=\frac{\lambda\Psi\otimes e^{2\rho}+e^{-2\rho}\otimes\lambda\Psi}{\Delta(\lambda)}, \\
&\Delta(Q_i)=Q_i\otimes e^{\rho}+e^{-\rho}\otimes Q_i, \qquad \Delta(P_i)=P_i\otimes e^{\rho}+e^{-\rho}\otimes P_i, \qquad (i=1,2),
\end{eqnarray*}
counit defined by
$$\epsilon(\Theta)=\epsilon(\Phi)=\epsilon(\Psi)=\epsilon(Q_1)=\epsilon(Q_2)=\epsilon(P_1)=\epsilon(P_2)=0,$$
and antipode defined by
\begin{eqnarray*}
&S(\Theta)=-\Theta, ~~ S(\Phi)=-\Phi, ~~ S(\Psi)=-\Psi, ~~ S(Q_i)=-Q_i, ~~ S(P_i)=-P_i, ~~ (i=1,2).
\end{eqnarray*}
Moreover, $U_{\hbar_1,\hbar_2,\hbar_3}(\G^{\alpha,\beta,\gamma})$ is a three-parameter quantization of $\G^{\alpha,\beta,\gamma}$.
\end{Prop}

\begin{proof}
(\textbf{$\Delta$ is a well-defined algebra homomorphism}) Note that $\rho, \lambda$ are central elements, we have $[Q_1\otimes e^\rho,e^{-\rho}\otimes P_1]=0$ and $$[Q_1\otimes e^\rho,P_1\otimes e^\rho]=Q_1 P_1\otimes e^{2\rho}-P_1 Q_1\otimes e^{2\rho}=[Q_1,P_1]\otimes e^{2\rho}=\frac{\lambda}{\alpha}\Theta\otimes e^{2\rho}.$$
From this we get
\begin{align*}
[\Delta(Q_1),\Delta(P_1)]&=[Q_1\otimes e^{\rho}+e^{-\rho}\otimes Q_1, P_1\otimes e^{\rho}+e^{-\rho}\otimes P_1] \\
&={1\over\alpha}(\lambda\Theta\otimes e^{2\rho}+e^{-2\rho}\otimes \lambda\Theta)=\Delta(\frac{\lambda}{\alpha}\Theta)=\Delta([Q_1,P_1]).
\end{align*}
Similarly, we can prove that $[\Delta(Q_2),\Delta(P_2)]=\Delta([Q_2,P_2])$, $[\Delta(Q_1),\Delta(Q_2)]=\Delta([Q_1,Q_2])$ and $[\Delta(P_1),\Delta(P_2)]=\Delta([P_1,P_2])$. Furthermore, since $\Delta(\Theta)$ is a central element, we have $$[\Delta(\Theta),\Delta(\Phi)]=\Delta([\Theta,\Phi])=[\Delta(\Theta),\Delta(P_1)]=\Delta([\Theta,P_1])=0.$$
We also get
$$[\Delta(P_1),\Delta(Q_2)]=[P_1\otimes e^{\rho}+e^{-\rho}\otimes P_1, Q_2\otimes e^{\rho}+e^{-\rho}\otimes Q_2]=0=\Delta([P_1,Q_2]).$$
Thus all typical defining relations of algebra are checked to be preserved by $\Delta$, hence the other defining relations are also preserved. We also have to check the consistency of the definitions of $\Delta(\Theta)$, $\Delta(\Phi)$, $\Delta(\Psi)$ and $\Delta(\rho)$. Indeed,
\begin{align*}
\Delta(\hbar_1 \Theta+\hbar_2 \Phi+\hbar_3 \Psi)&=\frac{\rho\lambda\otimes e^{2\rho}+e^{-2\rho}\otimes\rho\lambda}{\Delta(\lambda)} \\
&=\frac{(e^{2\rho}-e^{-2\rho})\otimes e^{2\rho}+e^{-2\rho}\otimes(e^{2\rho}-e^{-2\rho})}{\Delta(4\lambda)}=\Delta(\rho).
\end{align*}
Thus $\Delta$ is actually a well-defined algebra homomorphsim.

(\textbf{$\Delta$ is coassociative}) Note that $\Delta(\lambda)=\frac{e^{2\rho}\otimes e^{2\rho}-e^{-2\rho}\otimes e^{-2\rho}}{4\Delta(\rho)}$, we have $$(1\otimes\Delta)\Delta(\lambda)=\frac{e^{2\rho}\otimes e^{2\rho}\otimes e^{2\rho}-e^{-2\rho}\otimes e^{-2\rho}\otimes e^{-2\rho}}{4(\rho\otimes 1\otimes 1+1\otimes\rho\otimes 1+1\otimes 1\otimes\rho)}=(\Delta\otimes 1)\Delta(\lambda).$$
Therefore, to prove $(1\otimes\Delta)\Delta(\Theta)=(\Delta\otimes 1)\Delta(\Theta)$, we only need to check
$$(1\otimes\Delta)(\lambda\Theta\otimes e^{2\rho}+e^{-2\rho}\otimes \lambda\Theta)=(\Delta\otimes 1)(\lambda\Theta\otimes e^{2\rho}+e^{-2\rho}\otimes \lambda\Theta).$$
The two sides of the above equality are all easily seen to be $$\lambda\Theta\otimes e^{2\rho}\otimes e^{2\rho}+e^{-2\rho}\otimes \lambda\Theta\otimes e^{2\rho}+e^{-2\rho}\otimes\otimes e^{-2\rho}\otimes \lambda\Theta.$$
This verifies the coassociativity of $\Delta$ on $\Theta$. The coassociativity of $\Delta$ on the other generators can be verified similarly.

(\textbf{$\epsilon$ is a well-defined algebra homomorphism}) This is obvious.

(\textbf{$\epsilon$ is truely a counit}) By simple computation, we see $(1\otimes\epsilon)\Delta(\Theta)=\frac{\lambda\Theta}{\lambda}=\Theta$. Such relation for the other generators are verified by similar computation.

(\textbf{$S$ is a well-defined algebra anti-homomorphism}) Note that $S(\rho)=-\rho$ and $S(e^\rho)=e^{-\rho}$, we have $S(\lambda)=\lambda$. Let us check that $$[S(Q_1),S(P_1)]=[-Q_1,-P_1]={\lambda\over\alpha}\Theta=S([P_1,Q_1]).$$ The other defining relations are checked similarly.

(\textbf{The structure maps are compatible}) We want to prove that for each generator $X\in\{\Theta,\Phi,\Psi,Q_1,Q_2,P_1,P_2\}$, we get
$\mu(S\otimes 1)\Delta(X)=\mu(1\otimes S)\Delta(X)=\epsilon(X)=0$, where $\mu$ is the multiplication map from the algebra structure. Due to the similarity between generators, the invertibility of $\lambda$, and the ``almost symmetric'' nature of $\Delta$, we only prove the equalities $\mu(S\otimes 1)\Delta(\lambda\Theta)=\mu(S\otimes 1)\Delta(P_1)=0$. In fact, 
\begin{eqnarray*}
&\mu(S\otimes 1)\Delta(\lambda\Theta)=\mu(S\otimes 1)(\lambda\Theta\otimes e^{2\rho}+e^{-2\rho}\otimes \lambda\Theta)=-\lambda\Theta e^{2\rho}+e^{2\rho}\lambda\Theta=0, \\
&\mu(S\otimes 1)\Delta(P_1)=\mu(S\otimes 1)(P_1\otimes e^{\rho}+e^{-\rho}\otimes P_1)=-P_1 e^\rho+e^\rho P_1=0.
\end{eqnarray*}
This proves the compatibility of structure maps.

(\textbf{Classical limit coincides with $\G^{\alpha,\beta,\gamma}$}) Note that $\lambda\equiv 1 ~~ (\mathrm{mod}\: \hbar_1,\hbar_2,\hbar_3)$, we get
$$[Q_1,P_1]=[Q_2,P_2]\equiv{1\over\alpha}\Theta, \quad [Q_1,Q_2]\equiv\frac{\beta}{\alpha^2}\Phi, \quad [P_1,P_2]\equiv\frac{\gamma}{\alpha^2}\Psi, \quad (\mathrm{mod}\: \hbar_1,\hbar_2,\hbar_3)$$
which are the defining relations of $\G^{\alpha,\beta,\gamma}$. We also see that in classical limit, each generator is primitive, e.g., $\Delta(\Theta)\equiv \Theta\otimes 1+1\otimes \Theta$. Therefore, if we denote $U=U_{\hbar_1,\hbar_2,\hbar_3}(\G^{\alpha,\beta,\gamma})$, then $U/(\hbar_1 U+\hbar_2 U+\hbar_3 U)\cong U(\G^{\alpha,\beta,\gamma})$ as Hopf algebras.
\end{proof}

\begin{remark}
Note that $\G^{1,0,0}\cong\mathbb{C}^2\oplus\Gh$ is the trivial two-dimensional central extension of the Heisenberg algebra $\Gh$ in two dimensions. Taking $\hbar_2,\hbar_3\to 0$, we find that $U_{\hbar,0,0}(\G^{1,0,0})$ is defined as follows:
\begin{eqnarray*}
&[Q_i,P_j]=\delta_{ij}\frac{\sinh{(2\hbar \Theta)}}{2\hbar}, \qquad \Delta(\Theta)=\Theta\otimes 1+1\otimes \Theta, \\
&\Delta(Q_i)=Q_i\otimes e^{\hbar \Theta}+e^{-\hbar \Theta}\otimes Q_i, \quad \Delta(P_i)=P_i\otimes e^{\hbar \Theta}+e^{-\hbar \Theta}\otimes P_i, \quad (i=1,2).
\end{eqnarray*}
This is a quantization of Hisenberg algebra defined by E. Celeghini, et al. (see \cite{CeGiSoTa}).

Also note that in the Heisenberg algebra case, there exist coboundary but non-quasi-triangular Lie bialgebra structures; and if we would like to extend the Heisenberg algebra by a derivation, then there is a quasitriangular Lie bialgebra structure, and the above quantization provids an universal R-matrix, which can be obtained by a contraction on an R-matrix of the Drinfeld-Jimbo type quantum group (see \cite{CeGiSoTa}).

However, when $\alpha,\beta,\gamma \ne 0$, we checked by the mathematical software \textsf{Mathematica} that there is no nontrivial coboundary Lie bialgebra structure on $\G^{\alpha,\beta,\gamma}$. Thus there is no universal R-matrix for generic $\G^{\alpha,\beta,\gamma}$ in the ordinary sense.
\end{remark}

\begin{Prop}
The Hopf algebra $U_{\hbar_1,\hbar_2,\hbar_3}(\G^{\alpha,\beta,\gamma})$ is a flat deformation of $U(\G^{\alpha,\beta,\gamma})$. Namely, $U_{\hbar_1,\hbar_2,\hbar_3}(\G^{\alpha,\beta,\gamma})\cong U(\G^{\alpha,\beta,\gamma})[[\hbar_1,\hbar_2,\hbar_3]]$ as associative algebras, hence only the coalgebra structure of $U(\G^{\alpha,\beta,\gamma})$ is deformed. Here we use $U(\G^{\alpha,\beta,\gamma})[[\hbar_1,\hbar_2,\hbar_3]]$ to denote the trivial deformation of $U(\G^{\alpha,\beta,\gamma})$.
\end{Prop}

\begin{proof}
Let $\phi : U_{\hbar_1,\hbar_2,\hbar_3}(\G^{\alpha,\beta,\gamma})\to U_{\hbar_1,\hbar_2,\hbar_3}(\G^{\alpha,\beta,\gamma})$ be the algebra automorphism such that
\begin{eqnarray*}
&\phi(\Theta)=\Theta/\lambda, ~~ \phi(\Phi)=\Phi/\lambda, ~~ \phi(\Psi)=\Psi/\lambda, ~~ \phi(Q_i)=Q_i/\lambda, ~~ \phi(P_i)=P_i/\lambda, ~~ (i=1,2).
\end{eqnarray*}
Since
\begin{eqnarray*}
&[\phi(Q_i),\phi(P_j)]=\frac{\delta_{ij}\lambda}{\alpha\lambda^2}\Theta=\frac{\delta_{ij}}{\alpha}\phi(\Theta), \qquad (i,j=1,2),\\
&[\phi(Q_1),\phi(Q_2)]=\frac{\beta\lambda}{\alpha^2\lambda^2}\Phi=\frac{\beta}{\alpha^2}\phi(\Phi), \quad [\phi(P_1),\phi(P_2)]=\frac{\gamma\lambda}{\alpha^2\lambda^2}\Psi=\frac{\gamma}{\alpha^2}\phi(\Psi),
\end{eqnarray*}
we readily see that $\phi$ is well-defined. However, the above equalities are obviously the defining relations of $\G^{\alpha,\beta,\gamma}$, hence $U_{\hbar_1,\hbar_2,\hbar_3}(\G^{\alpha,\beta,\gamma})\cong U(\G^{\alpha,\beta,\gamma})[[\hbar_1,\hbar_2,\hbar_3]]$ as algebras.
\end{proof}

\begin{Cor}\label{pbw-basis}
The set $\left\{ \Theta^{e_1}\Phi^{e_2}\Psi^{e_3}Q_1^{e_4}Q_2^{e_5}P_1^{e_6}P_2^{e_7}~|~e_i\in\mathbb{N},\;i=1,\cdots,7 \right\}$ forms a topological basis of the free $\mathbb{C}[[\hbar_1,\hbar_2,\hbar_3]]$-module $U_{\hbar_1,\hbar_2,\hbar_3}(\G^{\alpha,\beta,\gamma})$. Here $\mathbb{N}$ is the set consisting of nonnegative integers.
\end{Cor}

\begin{proof}
This follows from the PBW theorem and the invertibility of $\lambda$.
\end{proof}

Now let us have a look at the Lie bialgebra structure on $\G^{\alpha,\beta,\gamma}$ associated with the quantization $U_{\hbar_1,\hbar_2,\hbar_3}(\G^{\alpha,\beta,\gamma})$. Denote $\Delta^o=\sigma\circ\Delta$, where $\sigma$ is the flip operator, namely, $\sigma(x\otimes y)=y\otimes x$. In $U_{\hbar_1,0,0}(\G^{\alpha,\beta,\gamma})$, we have
\begin{align*}
\frac{\Delta(\Theta)-\Delta^o(\Theta)}{\hbar_1}&\equiv\Theta\otimes\left(\frac{e^{2\hbar_1\Theta}-e^{-2\hbar_1\Theta}}{\hbar_1}\right) - \left(\frac{e^{2\hbar_1\Theta}-e^{-2\hbar_1\Theta}}{\hbar_1}\right)\otimes \Theta \\
&\equiv \Theta\otimes 4\Theta - 4\Theta\otimes\Theta \equiv 0 \quad (\mathrm{mod} \; \hbar_1).
\end{align*}
In $U_{0,\hbar_2,0}(\G^{\alpha,\beta,\gamma})$, we have
\begin{align*}
\frac{\Delta(\Theta)-\Delta^o(\Theta)}{\hbar_2}&\equiv\Theta\otimes\left(\frac{e^{2\hbar_2 \Phi}-e^{-2\hbar_2 \Phi}}{\hbar_2}\right) - \left(\frac{e^{2\hbar_2 \Phi}-e^{-2\hbar_2 \Phi}}{\hbar_2}\right)\otimes\Theta \\
&\equiv \Theta\otimes 4\Phi - 4\Phi\otimes\Theta \equiv 4\Theta\wedge\Phi \quad (\mathrm{mod} \; \hbar_2).
\end{align*}
And in $U_{0,0,\hbar_3}(\G^{\alpha,\beta,\gamma})$, we get
\begin{align*}
\frac{\Delta(\Theta)-\Delta^o(\Theta)}{\hbar_3}&\equiv\Theta\otimes\left(\frac{e^{2\hbar_3 \Psi}-e^{-2\hbar_3 \Psi}}{\hbar_3}\right) - \left(\frac{e^{2\hbar_3 \Psi}-e^{-2\hbar_3 \Psi}}{\hbar_3}\right)\otimes\Theta \\
&\equiv \Theta\otimes 4\Psi - 4\Psi\otimes\Theta \equiv 4\Theta\wedge\Psi \quad (\mathrm{mod} \; \hbar_3).
\end{align*}
This implies that the cocommutator $\delta$ of $\G^{\alpha,\beta,\gamma}$ satifies $$\delta(\Theta)=2\Theta\wedge(a \Theta+b \Phi+c \Psi)$$
for some fixed $a,b,c\in\mathbb{C}$. By a similar computation, we get
\begin{eqnarray*}
&\delta(\Phi)=2\Phi\wedge(a\Theta+b\Phi+c\Psi), \qquad \delta(\Psi)=2\Psi\wedge(a\Theta+b\Phi+c\Psi), \\
&\delta(Q_i)=Q_i\wedge(a\Theta+b\Phi+c\Psi), \quad \delta(P_i)=P_i\wedge(a\Theta+b\Phi+c\Psi), \quad (i=1,2).
\end{eqnarray*}
Clearly this is not a coboundary Lie bialgebra structure when $a,b,c$ are not all zero, since $\Theta$, $\Phi$ and $\Psi$ are central elements.

\section{Star-product on $C^\infty(\g^{\alpha,\beta,\gamma})$}\label{sec:star-prod}
In this section, we turn to the dual perspective, namely, working out the $*$-product on $C^\infty(\g^{\alpha,\beta,\gamma})$ (the set of smooth functions on $\g^{\alpha,\beta,\gamma}$). Being endowed with the $*$-product and the original coproduct, the set $C^\infty(\g^{\alpha,\beta,\gamma})$ turns out to be a non-commutative and non-cocommutative Hopf algebra, which is denoted by $C_{\hbar_1,\hbar_2,\hbar_3}^\infty(\g^{\alpha,\beta,\gamma})$. To simplify the problem, we only consider in this paper a local counterpart of $C_{\hbar_1,\hbar_2,\hbar_3}^\infty(\g^{\alpha,\beta,\gamma})$, namely, $U_{\hbar_1,\hbar_2,\hbar_3}(\G^{\alpha,\beta,\gamma})^\ast$, the dual Hopf algebra of $U_{\hbar_1,\hbar_2,\hbar_3}(\G^{\alpha,\beta,\gamma})$.

For convenience, let us introduce some useful notations. Suppose $I=(i_1,\cdots,i_s)$ is an $s$-tuple, then $i_1+i_2+\cdots+i_s$ is denoted by $|I|$, and the expression $i_1!i_2!\cdots i_s!$ is abrreviated as $I!$. If $J=(j_1,\cdots,j_s)$ is another tuple, we also abrreviate $\binom{i_1}{j_1}\binom{i_2}{j_2}\cdots\binom{i_s}{j_s}$ to $\binom{I}{J}$, where $\binom{i}{j}=\frac{i!}{j!(i-j)!}$ are binomial coefficients. Note that $\frac{|I|!}{I!}$ is denoted by $\binom{|I|}{I}$, $\frac{(\lambda\Theta)^{i_1}(\lambda\Phi)^{i_2}(\lambda\Psi)^{i_3}}{I!}$ is abrreviated to $Z^I$, $\frac{Q_1^{i_1}Q_2^{i_2}P_1^{i_3}P_2^{i_4}}{I!}$ is abrreviated to $X^I$, and $\hbar_1^{i_1}\hbar_2^{i_2}\hbar_3^{i_3}$ is shortened as $H^I$. Finally, given two $s$-tuples $I$ and $J$, we can form their linear combination $aI+bJ=(ai_1+bj_1,ai_2+bj_2,\cdots,ai_s+bj_s)$ for any $a,b\in\mathbb{Z}$. 

It can be readily deduced from the Corollary \ref{pbw-basis} that $\{Z^I X^J~|~I\in\mathbb{N}^3, J\in\mathbb{N}^4\}$ is a topological basis for the free $\mathbb{C}[[\hbar_1,\hbar_2,\hbar_3]]$-module $U_{\hbar_1,\hbar_2,\hbar_3}(\G^{\alpha,\beta,\gamma})$. Let $\{W^K Y^L\}\in U_{\hbar_1,\hbar_2,\hbar_3}(\G^{\alpha,\beta,\gamma})^\ast$ be the dual basis, i.e., $\langle W^K Y^L, Z^I X^J\rangle=\delta_{KI}\delta_{LJ}$. Then we only need to compute the $*$-product $W^I Y^J * W^K Y^L$ by using the identity
$$\langle W^I Y^J * W^K Y^L, Z^S X^T \rangle=\langle W^I Y^J\otimes W^K Y^L, \Delta(Z^S X^T) \rangle.$$
In fact, we have
\begin{align*}
\Delta(Z^S X^T)&=\sum_{0\le I\le S,\: 0\le J\le T}Z^I X^J e^{-(2|S-I|+|T-J|)\rho}\otimes Z^{S-I} X^{T-J} e^{(2|I|+|J|)\rho} \\
&=\sum_{0\le I\le S,\: 0\le J\le T,\: M\ge 0,\: N\ge 0} \binom{I+M}{M}\binom{S-I+N}{N} H^{M+N}\times \\
& \quad \left(-2|S-I|-|T-J|\right)^{|M|} \left(2|I|+|J|\right)^{|N|} Z^{I+M} X^J\otimes Z^{S-I+N} X^{T-J}  \\
&=\sum_{0\le J\le T,\: 0\le M\le I,\: 0\le N\le K,\: (I-M)+(K-N)=S} \binom{I}{M}\binom{K}{N} H^{M+N} \times \\
& \quad \left(-2|K-N|-|T-J|\right)^{|M|} \left(2|I-M|+|J|\right)^{|N|} Z^I X^J\otimes Z^K X^{T-J}.
\end{align*}
It follows that
\begin{eqnarray*}
& W^I Y^J * W^K Y^L = \sum\limits_{0\le M\le I,\: 0\le N\le K} H^{M+N}W^{I+K-M-N} Y^{J+L} \times \\
& \quad \binom{I}{M}\binom{K}{N} \left(-2|K-N|-|L|\right)^{|M|} \left(2|I-M|+|J|\right)^{|N|}.
\end{eqnarray*}
As a formal power series in $\hbar_1$, $\hbar_2$ and $\hbar_3$, the constant term of $W^I Y^J * W^K Y^L$ is clearly $W^{I+K} Y^{J+L}$, thus the classical limit of the $*$-product is indeed commutative. Moreover, the linear term of $W^I Y^J * W^K Y^L$ is
\begin{eqnarray*}
&\hbar_1\left[k_1(2|I|+|J|)-i_1(2|K|+|L|)\right] W^{I+K-(1,0,0)} Y^{J+L} \\
+&\hbar_2\left[k_2(2|I|+|J|)-i_2(2|K|+|L|)\right] W^{I+K-(0,1,0)} Y^{J+L} \\
+&\; \hbar_3\left[k_3(2|I|+|J|)-i_3(2|K|+|L|)\right] W^{I+K-(0,0,1)} Y^{J+L},
\end{eqnarray*}
hence the Poisson bracket $\{W^I Y^J , W^K Y^L\}$ associated with the classical limit of the $*$-product is
\begin{eqnarray*}
&a \left[k_1(2|I|+|J|)-i_1(2|K|+|L|)\right] W^{I+K-(1,0,0)} Y^{J+L} \\
+&b \left[k_2(2|I|+|J|)-i_2(2|K|+|L|)\right] W^{I+K-(0,1,0)} Y^{J+L} \\
+&c \left[k_3(2|I|+|J|)-i_3(2|K|+|L|)\right] W^{I+K-(0,0,1)} Y^{J+L}
\end{eqnarray*}
for some fixed $a,b,c\in\mathbb{C}$. In particular, if we denote
\begin{eqnarray*}
& W^{(1,0,0)}=\chi_1, \qquad W^{(0,1,0)}=\chi_2, \qquad W^{(0,0,1)}=\chi_3, \\
& Y^{(1,0,0,0)}=\chi_4, \quad Y^{(0,1,0,0)}=\chi_5, \quad Y^{(0,0,1,0)}=\chi_6, \quad Y^{(0,0,0,1)}=\chi_7, 
\end{eqnarray*}
then we get
\begin{eqnarray*}
&\{\chi_1,\chi_2\}=2(b\chi_1-a\chi_2), \qquad \{\chi_1,\chi_3\}=2(c\chi_1-a\chi_3), \qquad \{\chi_2,\chi_3\}=2(c\chi_2-b\chi_3), \\
&\{\chi_i,\chi_1\}=a \chi_i, \quad \{\chi_i,\chi_2\}=b \chi_i, \quad \{\chi_i,\chi_3\}=c \chi_i, \quad \{\chi_i,\chi_j\}=0, \qquad (4\le i,j\le 7).
\end{eqnarray*}
This is a linear Poisson structure, thus defines a seven-dimensional Lie algebra, which is, as expected, the Lie algebra structure on $(\G^{\alpha,\beta,\gamma})^\ast$ corresponding to the cocommutator $\delta$.

\section{Conclusion}\label{sec:conclusions}
In this paper we studied the deformation of noncommutative quantum mechanics (NCQM) in two dimensions by quantizing its kinematical symmetry group $\g$ employing techniques of quantum groups. We want to re-emphasize in this context that the authors in \cite{Diasetal} outlined a deformation quantization scheme by introducing new $*$-products between elements of the Schwartz space of test functions and its dual on $\mathbb{R}^{2n}$.

To this end, they suitably deformed the standard symplectic matrix $J$ to obtain the skew-symmetric matrix $\Omega$ which approaches $J$ as $\hbar\rightarrow 0$. In \cite{wigjmp}, on the other hand, a three-parameter family of $*$-products, between Wigner functions of $\g$ supported on its various coadjoint orbits, was introduced in terms of which properties of noncommutative marginal distributions in position and momentum coordinates as obtained in \cite{bastoscmp} were also verified to hold. The $*$-products and the relevant Poisson structures that we have obtained in this paper are drastically of different nature as they are the deformed products for the algebra of functions on $\g$, or more generally on $\g^{\alpha,\beta,\gamma}$.

Also it is important to note that not all of the algebras in the family $\G^{\alpha,\beta,\gamma}$ are non-isomorphic. We have partial results on this classification problem which we propose to study in greater depth in a future publication.

\section*{Acknowledgements}
The authors would like to thank Professor Chengming Bai for great encouragement and Professor Syed Twareque Ali for beneficial discussions.

One of the authors (SHHC) gratefully acknowledges a grant from National Natural Science Foundation of China (NSFC) under Grant No. 11550110186.


\end{document}